\documentclass[conference]{IEEEtran}
\IEEEoverridecommandlockouts
\usepackage{cite}
\usepackage{amsmath,amssymb,amsfonts}
\usepackage{algorithmic}
\usepackage{soul}
\usepackage{graphicx}
\usepackage{textcomp}
\usepackage{xcolor}
\usepackage{enumitem}
\usepackage{subfig}
\usepackage{listings}

\usepackage[normalem]{ulem}
\usepackage{xspace}

\usepackage{diagbox}
\usepackage{tabularray}
\usepackage{enumitem}
\usepackage{mathrsfs} 
\usepackage{multirow}
\usepackage{booktabs}
\usepackage{algorithm} 
\usepackage{xurl}

\newtheorem{proof}{\textbf{Proof}} 
\newtheorem{corollary}{Corollary}
\IEEEoverridecommandlockouts\IEEEpubid{\makebox[\columnwidth]{ 979-8-3503-5171-2/24/\$31.00 $\copyright$2024 IEEE \hfill}\hspace{\columnsep}\makebox[\columnwidth]{ }}    
    
\begin{document}


\setlength{\paperheight}{11in}
\setlength{\paperwidth}{8.5in}

\newcommand{\name}{{\sc StarVeri}\xspace}
\newcommand{\etc}{\emph{etc.}}
\newcommand{\ie}{\emph{i.e.,}}
\newcommand{\eg}{\emph{e.g.,}}
\newcommand{\etal}{\emph{et al.}}

\newcommand{\citeColored}[1]{{\hypersetup{citecolor=blue}\cite{‌​#1}}}

\newenvironment{itmz}{\vspace{-\topsep}\begin{itemize}[leftmargin=*]\parskip0pt\partopsep0pt\itemsep3pt\topsep0 pt}{\end{itemize}\vspace{-\topsep}}

\newenvironment{enum}{\vspace{-\topsep}\begin{enumerate}[leftmargin=*]\parskip0pt\partopsep0pt\itemsep3pt\topsep0 pt}{\end{enumerate}\vspace{-\topsep}}

\newenvironment{desc}{\begin{description}\parskip0pt\partopsep0pt\itemsep3pt\topsep0pt}{\end{description}}

\newcommand{\cmark}{\ding{51}}%
\newcommand{\xmark}{\ding{55}}%

\title{\name: Efficient and Accurate Verification for Risk-Avoidance Routing in LEO Satellite Networks}

\author{
	\IEEEauthorblockN{Chenwei Gu$^{\dag}$, Qian Wu$^{\dag\ddag}$, Zeqi Lai$^{\dag\ddag\textsuperscript{${^\P}$}\thanks{${^\P}$~Zeqi Lai is the corresponding author.}}$, Hewu Li$^{\dag\ddag}$, Jihao Li$^{\ddag}$, Weisen Liu$^{\dag}$, Qi Zhang$^{\ddag}$, Jun Liu$^{\dag\ddag}$, Yuanjie Li$^{\dag\ddag}$}
	\IEEEauthorblockA{
		\textit{$^{\dag}$Institute for Network Sciences and Cyberspace, Tsinghua University, Beijing 100084, China} \\
		\textit{$^{\ddag}$Zhongguancun Laboratory, Beijing, China}\\
	}
} 

\maketitle

\begin{abstract}
	
Emerging satellite Internet constellations such as SpaceX's Starlink will deploy thousands of broadband satellites and construct Low-Earth Orbit~(LEO) satellite networks~(LSNs) in space, significantly expanding the boundaries of today's terrestrial Internet. However, due to the unique global LEO dynamics, satellite routers will inevitably pass through uncontrolled areas, suffering from security threats. It should be important for satellite network operators~(SNOs) to enable \emph{verifiable risk-avoidance routing} to identify path anomalies. In this paper, we present \name, a novel network path verification framework tailored for emerging LSNs. \name addresses the limitations of existing crypto-based and delay-based verification approaches and accomplishes \emph{efficient and accurate path verification} by: (i) adopting a dynamic relay selection mechanism deployed in SNO's operation center to judiciously select verifiable relays for each communication pair over LSNs; and (ii) incorporating a lightweight path verification algorithm to dynamically verify each segment path split by distributed relays. We build an LSN simulator based on real constellation information and the results demonstrate that \name can significantly improve the path verification accuracy and achieve lower router overhead compared with existing approaches.
\end{abstract}
\section{Introduction}
\label{sec:introduction}
Low-Earth Orbit~(LEO) Satellite Networks~(LSNs) are gaining tremendous popularity in recent years, carrying a large fraction of Internet traffic~\cite{traffic_boom_news}. Many ``NewSpace'' players like SpaceX~\cite{starlink} and Amazon Kuiper Project~\cite{kuiper} are constructing their own LSNs powered by laser inter-satellite links~(ISLs)~\cite{fcc1,llc} to provide global Internet services. In particular, Starlink, the largest operational LSN today, has launched more than 6300 LEO satellites~\cite{Jonathan} and attracted more than 3 million global subscribers as of May 2024~\cite{new-pop}.


As LSNs are developing at such a fast pace, security issues have become increasingly prominent. Potential threats in LSNs include not only traditional electromagnetic interference~\cite{twitter-news}, but also new threats in cyberspace since LSNs target to provide global Internet services. Fundamentally, LSNs have a unique characteristic different from the terrestrial Internet: the core network infrastructures in an LSN~(\eg\,a large number of satellite routers) are moving in their free-space orbits globally, and satellites will inevitably travel to uncontrolled and risky regions overseas. Such LEO dynamics can involve potential security risks such as traffic hijacking and information leakages~\cite{info_steal} by hacking satellites~\cite{attack1,attack2,news-hijack1}. Therefore, for satellite network operators~(SNOs), it should be important to enforce the forwarding path to bypass potential risk areas and, more importantly, verify that the actual forwarding paths do avoid the risk areas.

The network community has a very long history of working on path or routing verification. Depending on the verification criteria, previous approaches can be classified into two main categories: crypto-based and delay-based path verification. In the former, each intermediate node performs cryptographic operations on forwardings packets and appends a Message Authentication Code~(MAC) to each packet~\cite{CoNext11-pv,SIGCOMM14-pv,sec20-pv}. Hence, the path can be verified by checking a chain consisting of nested MACs. However, such complex cryptographic operations incur substantial computation and bandwidth overheads on resource-constrained satellite routers.
	 

To avoid high router overhead, other efforts have proposed lightweight delay-based approaches to verify the network paths by comparing the estimated delay of the planned path with the real delay of the actual path~\cite{SIGCOMM15-pv,sec17-pv,NDSS19-pv}. In delay-based approaches, the network operator typically has to pre-deploy a set of \emph{relay nodes} in advance, which are geographically very far away from the risk area. The path from a source to its destination is forced to pass through a certain relay. The underlying verification principle is that: due to the long distance between the relay and the risk area, if the actual path detours and passes through the risk area, the end-to-end path delay should be significantly higher than the estimated value. Note that this approach estimates the path delay based on a fundamental assumption in today's terrestrial Internet that the propagation delay between two network nodes can be approximately estimated by a linear function of the physical distance between them~\cite{SIGCOMM09-liner_relation}. However, in emerging LSNs, due to the unique global LEO dynamics, the entire network topology and paths fluctuate frequently~\cite{HotNets18-routing,mobicom24-routing}, causing the linear assumption does not hold anymore. As a result, existing delay-based methods suffer from poor accuracy in LSNs.

To overcome the inefficiency and inaccuracy of existing verification approaches, in this paper, we present \name, a novel verification framework that extends existing efforts and accomplishes efficient and accurate network path verification in dynamic LSNs, by exploiting the following two techniques.

First, \name incorporates a dynamic relay selection algorithm together with a flexible relay-based traffic steering mechanism for SNOs to dynamically schedule LSN traffic to bypass risk areas. Unlike previous approaches that depend on static relays which may result in high verification errors, \name proposes the \emph{Nearest Low-Risk Planes}~($NLRP$) to limit the risk nodes in a small range, and dynamically selects relays based on $NLRP$.
As such, \name effectively avoids the risk areas without involving too much delay caused by risk-bypassing meandering routes.

Second, \name adopts a lightweight avoidance verification algorithm that integrates the routing information and propagation delays to jointly verify the path compliance between the planned and the actual paths. \name verifies the entire network path by verifying all segments divided by relays. Specifically, in \name, only relays perform MAC operations to authenticate that the packets indeed pass through them and conduct delay-based verification for each segment path. Instead of using the linear RTT estimation like~\cite{SIGCOMM15-pv,sec17-pv,NDSS19-pv}, \name first adopts an inter-relay probing mechanism to obtain each segment delay ground truth for a period. Then \name estimates the segment detour delay threshold that functions as a jitter buffer by computing twice the minimum propagation delay from each node on the segment to the risk nodes based on the predictable satellite trajectory and topology \cite{CoNext23-bg}. The sum of them is the segment delay upper bound for determining if packets have traversed the risk area in this segment path.

To evaluate the effectiveness of \name, we conduct a large-scale simulation by combining real-world LSN information~\cite{starlink,kuiper} and the recent LSN simulator~\cite{StarryNet}. Extensive experiments demonstrate that \name can accomplish: (i) \emph{high verification accuracy:} \name obtains near-to-$100\%$ verification accuracy for city pairs served by Starlink and Kuiper constellations;
(ii) \emph{low delay penalty caused by detours:} compared with Alibi Routing~\cite{SIGCOMM15-pv}, \name can  largely reduce the delay of verifiable risk-avoidance paths; and (iii) \emph{better scalability and performance:} \name achieves low processing overhead and can scale to LSNs with thousands of nodes, and \name's achievable goodput ratio is much higher than existing crypto-based approaches~\cite{CoNext11-pv,SIGCOMM14-pv,sec20-pv}.



In summary, the main contributions of this paper include:

\begin{itemize}[leftmargin=*]
	\item We highlight the importance of path verification in emerging LSNs, and expose the inefficiency and inaccuracy problems of existing path verification approaches in LSNs~(\S\ref{sec:background}).
	
	\item We present \name, a novel path verification framework tailored for LSNs. \name exploits a dynamic relay-based traffic steering mechanism and a lightweight, segment avoidance verification algorithm to efficiently and accurately verify dynamic network paths in LSNs~(\S\ref{sec:design_overview},\S\ref{sec:design_details}).
	
	\item We build a \name prototype and conduct extensive simulations to demonstrate the effectiveness of \name, in terms of improved verification accuracy, performance, and reduced router overheads~(\S\ref{sec:evaluation}).
\end{itemize}

\section{Threat Model and motivation}
\label{sec:background}

\subsection{Preliminaries for low-Earth orbit satellite networks}
\label{subsec:leo_satellite_network}

\begin{figure}[tb]
	\centering
	\includegraphics[width=0.92\linewidth]{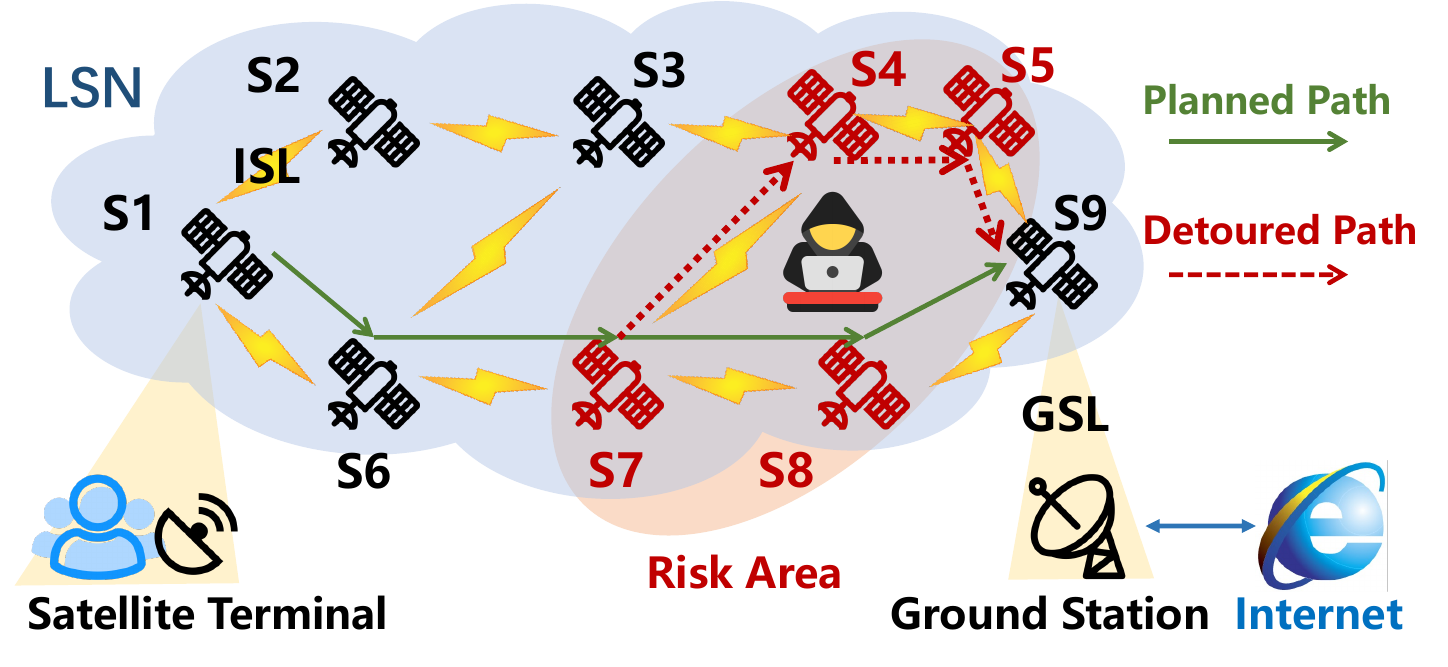}
	\caption{Routing security threats in an LSN: attackers in uncontrolled risk areas may steal information or hijack traffic.}
	\label{fig:LSN_architecture}
	\vspace{-0.3in}
\end{figure}

Today's LSNs like SpaceX's Starlink~\cite{starlink} consist of hundreds to thousands of LEO satellites that work as ``routers in space'', together with many geo-distributed ground stations connecting satellites and terrestrial network infrastructures. Satellites communicate with ground entities~(\eg\,ground stations or satellite terminals) via ground-satellite radio links~(GSL). Besides, many satellite constellations~(\eg\,Starlink and Kuiper) leverage high-speed laser inter-satellite links~(ISLs) for inter-satellite networking and communication. As visualized in \cite{starlink_map}, because satellites move at a high velocity in various orbital directions over the Earth, the entire LSN experiences frequent topology changes, especially in space-ground connections.


LSN traffic between two terrestrial nodes~(\eg\,from a Starlink's satellite terminal to a remote ground station) is forwarded by a network path constructed by uplink/downlink and ISLs. To deal with the unique LSN topology fluctuation, many recent works have proposed space routing mechanisms~(\eg~\cite{HotNets18-routing,mobicom24-routing}) to dynamically build and maintain paths for any two terrestrial nodes connected to the LSN.

\subsection{Potential risks in uncontrolled areas}
\label{subsec:routing_attack_risks}


As plotted in Fig.~\ref{fig:LSN_architecture}, because satellites move globally, they might enter an uncontrolled \emph{risk area}, posing \emph{routing security threats} in LSNs. In this paper, we define a ``risk area'' as a geographical region where malicious attackers may exist there and launch the following attacks on satellites above the area:

\begin{itemize}[leftmargin=*]

\item \textbf{Information stealing.} Attackers in the risk area can eavesdrop on traffic and extract or speculate private data~\cite{info_steal}. As shown in Fig.~\ref{fig:LSN_architecture}, assume traffic from the satellite terminal is expected to be forwarded to the ground station by a planned network path marked by the solid arrows, attackers in risk areas can eavesdrop on traffic forwarded from S7 to S9.

\item \textbf{Traffic hijacking.} More powerful attackers in a risk area can hijack \cite{attack1,attack2,news-hijack1} the route and cause path inconsistency. They propagate error routing advertisements to allure surrounding satellites to forward packets to them and redirect traffic to specific nodes for censorships, or other man-in-the-middle attacks like packet injection, modification, and counterfeit. In the example illustrated in Fig.~\ref{fig:LSN_architecture}, the attacker in the risk area may modify the planned path to S7$\rightarrow$S4$\rightarrow$S5$\rightarrow$S9.



\end{itemize}

Therefore, to avoid the above risks, it should be essential for satellite network operators~(SNOs) to: (i) apply avoidance policies to force their traffic to bypass potential risk areas, and more importantly,  (ii) \emph{verify} that the actual paths are consistent to the planned avoidance policies in practice.

\subsection{Are existing path verification methods sufficient?}
\label{subsec:limitation}

Over the past decade, the network community has had a body of efforts working on network path verification. In particular, existing path verification efforts can be classified into two main categories: \emph{crypto-based} and \emph{delay-based} approaches.

\noindent
\textbf{Crypto-based path verification.} One classic path verification approach is to embed the planned path~(\eg\,a sequence of nodes that build the network path) into the header of each packet. Then intermediate nodes authenticate and update related fields before sending to the next node~\cite{CoNext11-pv,SIGCOMM14-pv,sec20-pv}. These approaches have three limitations in LSNs. First, each intermediate node needs to perform cryptographic operations like calculating MACs hop by hop, which involves high computation overhead that could be unaffordable for resource-constrained satellites. Second, the increased length of the verification header also leads to extra packet header overhead, resulting in goodput ratio reduction~(as will be analyzed in detail in \S\ref{sec:evaluation}). Third, per-hop authentication inevitably involves additional processing delay on each node. Note that the high mobility of LEO satellites incurs frequent path changes~\cite{INFOCOM22-bg,CoNext23-bg,HotNets18-routing}, if the packet processing delay is too high, the pre-calculated path might be invalid on the route.
Although some recent works like PPV~\cite{IWQoS18-pv} and MASK~\cite{ToN23-pv} propose probabilistic authentication instead of verifying every packet hop by hop, the computation overhead can be still high when the traffic volume is large.





\noindent
\textbf{Delay-based path verification.} Alibi Routing~\cite{SIGCOMM15-pv} proposes a new path verification approach, exploiting the key idea that a detour from the planned path to any node in the risk area will breach the maintained delay by incurring significant extra delay. To verify whether a network path from a source ($Src$) to its destination ($Dst$) passes through the risk area, Alibi Routing pre-computes a \emph{target area} that is far away from the risk area where possible verifiable relays~(called alibis) are located. It enforces that the path from $Src$ to $Dst$ must pass through the alibis. Because the relay is very far from the risk area, if the $Src\rightarrow Dst$ path passes through the risk area, the observed end-to-end delay should be significantly higher than the maintained value. 
However, these approaches are based on a fundamental assumption that the delay between two terrestrial nodes has a linear relationship with their great-circle distance. Although this assumption exists in many scenarios in today’s terrestrial Internet~\cite{SIGCOMM09-liner_relation}, it is not applicable in LSNs with highly dynamic topology. An increase in path delay may not necessarily be caused by traversing the risk area but \emph{path changes} due to topology fluctuation.

We conduct a quantitative analysis to demonstrate how the unique topology fluctuations in LSNs affect the verification accuracy of existing delay-based approaches. Our analysis is based on the real Starlink constellation information~\cite{starlink} and an LSN simulation based on StarryNet~\cite{StarryNet}~(details will be illustrated in \S\ref{sec:evaluation}). We build a virtual LSN consisting of 1584 satellites with 72 orbits and 22 satellites per orbit~(Starlink Phase 1, Shell 1) and 165 geo-distributed ground stations~\cite{starlink-gs} around the world. We simulate about 2000 city pairs communicating over the LSN and use the shortest path routing upon the +Grid topology,\,\ie\,a satellite connects two neighboring satellites in the same orbit and two in the adjacent orbits~\cite{HotNets18-routing}.

\begin{figure}[t]
	\centering
	\setlength{\belowcaptionskip}{-20pt}
	\begin{minipage}[t]{0.47\linewidth} 
		\centering{
			\includegraphics[width=\textwidth]{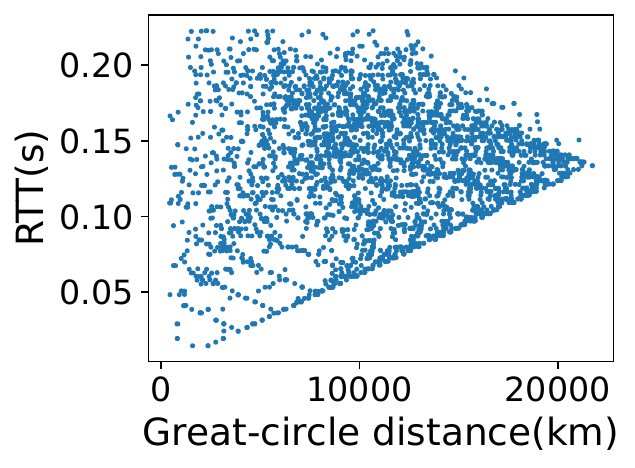}}
		\caption{Non-linear relationship between the great circle distance and RTTs.}
		\label{fig:nonlinear_relation}
	\end{minipage}
	\hspace{1mm}
	\begin{minipage}[t]{0.47\linewidth}
		\centering{
			\includegraphics[width=\textwidth]{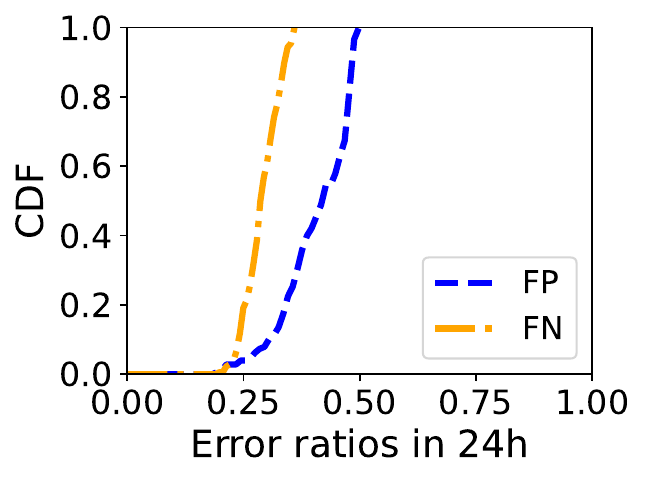}}
		\caption{Verification inaccuracy of Alibi Routing in dynamic LSNs.}
		\label{fig:error_ratio}
	\end{minipage}%
	\hspace{0.5mm}
	\vspace{0.1in}
\end{figure}
\noindent

First, we observe that the linear relationship assumption~\cite{SIGCOMM09-liner_relation} \emph{does not hold} in LSNs as illustrated in Fig.\ref{fig:nonlinear_relation}. On our further investigation, we confirm that the delay increase is caused by frequent topology fluctuations and path changes in LSNs even if the $Src$ and $Dst$ are static on the ground.
Second, we observe the non-linear relationship can lead to inaccuracy for delay-based verification approaches. We simulate the Alibi Routing~\cite{SIGCOMM15-pv} mechanism, set Egypt as a potential risk area, and set a static ground relay for each city pair. We calculate the ratio of false positive~(FP, \ie\,the real path does not traverse the risk area but is considered as passing through it) and false negative~(FN, \ie\,the opposite of FP) of each city pair in 24 hours with an interval of 1 second. These error ratios are defined as the proportion of the amount of time slots that FP or FN occurs to the total amount of time slots. Fig.~\ref{fig:error_ratio} plots the CDF of the error ratios of the city pairs. We observe that delay-based verification methods suffer from high inaccuracy since most city pairs' FN and FP ratios are higher than 25\%.

\noindent
\textbf{Our motivation.} Accomplishing path verification is important in dynamic and uncontrolled LSN environments. However, existing methods suffer from either per-node high processing overheads, or verification inaccuracy, which motivates us to explore a more efficient and accurate solution for network path verification in LSNs, \ie\,the \name framework.

\section{\name Overview}
\label{sec:design_overview}


\subsection{\name architecture}
\label{sec:key ideas}

\noindent
\textbf{Basic idea.} At a high level, \name addresses the limitations of existing approaches in terms of overhead and accuracy as follows. First, to verify a network path from $Src$ to $Dst$, \name adopts a collection of \emph{dynamic satellite relays} to split the entire path into a series of consecutive \emph{segments}. Thus, verifying the entire path is equivalent to verifying each segment path. Second, to reduce the verification overhead, \name only requires relays~(instead of all nodes on the path) to authenticate forwarded packets. Therefore, it can be verified that the forwarding path indeed passes through the predefined relays by checking the relay's authentication chain in $Dst$. Third, while the end-to-end delay of a network path fluctuates due to the LSN's dynamics, the delay fluctuation of a segment~(\eg\,between two adjacent relays)
is less dramatic than the delay fluctuation of the entire path. \name verifies each segment by comparing the real inter-relay delay with the estimated upper bound.


\noindent
\textbf{Architecture.} Fig.~\ref{fig:architecture_and_comparison} plots \name's high-level architecture and an illustration of the key differences between \name and existing verification approaches. \name provides the SNO with the ability to verify the path compliance between the planned packet delivery path and the real one. \name can be built upon existing Path-Aware Networking~(PAN)~\cite{ICNP20-SCIONLAB,springer-SCION} architecture that allows end users to self-define paths for packets, or source routing~\cite{SR}, both enforcing network paths to pass through pre-calculated satellite relays. In addition, \name incorporates two new features for dynamic LSN path verification. First, \name adopts a Dynamic Relay Selection mechanism deployed in SNO's operation center to judiciously select relays for every communication pair over the LSN. These relays then split the entire path from $Src$ to $Dst$ into a sequence of consecutive segments. Second, \name incorporates a probing mechanism to periodically obtain the segment delay between any two adjacent key nodes ($sat_s$, relays, and $sat_d$) and an Inter-Relay Verification approach deployed on selected satellite relays and $Dst$ to verify whether the segment path is consistent with the planned path. If the packet on the current segment path does not traverse the risk area, the relay node will perform a symmetric MAC to update the passing packet's authentication fields.



\begin{figure}[t]
	\vspace{-0.1in}
	\centering
	\subfloat[Crypto-based verification: high processing overhead on all satellites.]{
		\begin{minipage}[t]{0.46\linewidth}
			\centering{
				\includegraphics[width=\textwidth]{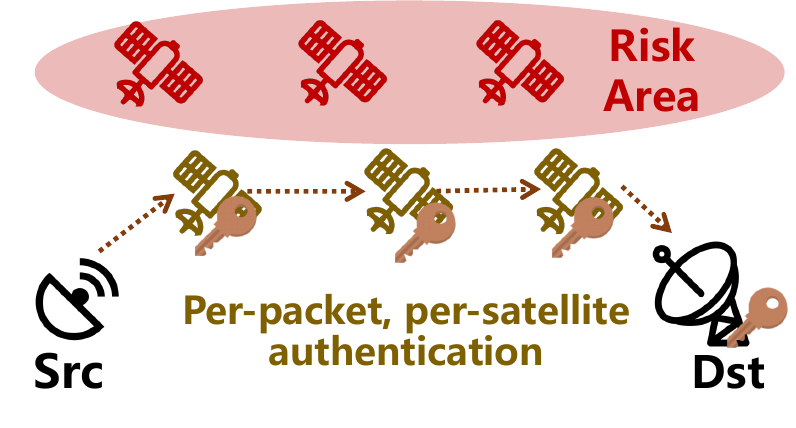}}
		\end{minipage}%
		
		\label{fig:crypto}
	}
	\hfill
	\subfloat[Delay-based verification: path fluctuations result in limited accuracy.]{
		\begin{minipage}[t]{0.46\linewidth}
			\centering{
				\includegraphics[width=\textwidth]{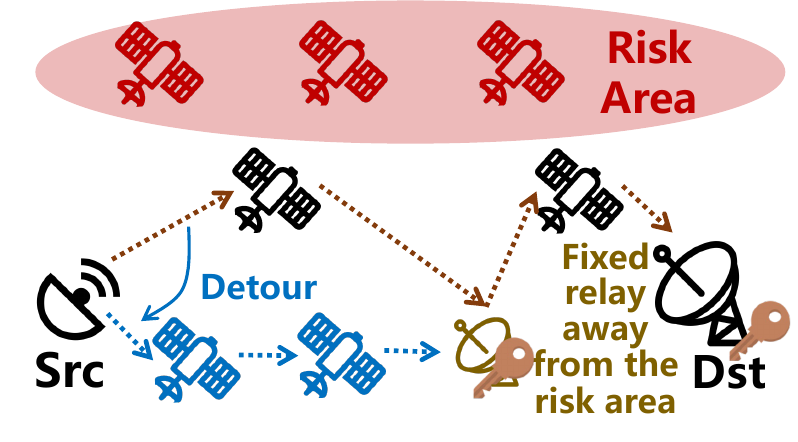}}
		\end{minipage}%
		\label{fig:delaybased}
	}
	\\
	\subfloat[\name performs crypto-based verification on relays, and extends existing delay-based verification approach to verify each segment of the path.]{
		\begin{minipage}[t]{0.99999\linewidth}
			\centering{
				\vspace{-0.1in}
				\includegraphics[width=\textwidth]{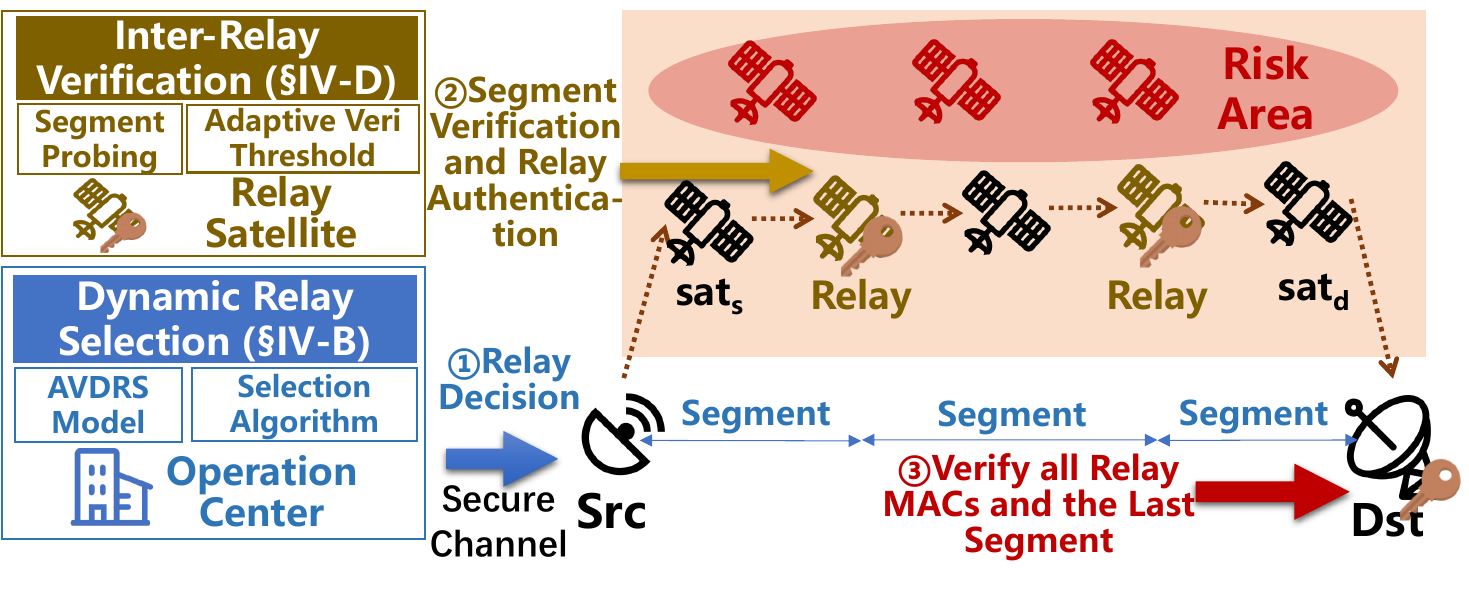}}
			\vspace{-0.2in}
			\label{fig:starveri_archi}
	 	\end{minipage}%
	}
	\caption{\name architecture and an illustration of the key differences from previous verification approaches.} 
	\vspace{-0.25in}
	\label{fig:architecture_and_comparison}
\end{figure}

\subsection{\name verification workflow}

For each communication pair ($Src$, $Dst$) and a given risk area, \name accomplishes the verification as follows.

\noindent
\textbf{(1) Dynamic relay selection.} \name first calculates a time-varying relay set which includes all relays for ($Src$, $Dst$) at different time slots. Then it estimates the segment detour delay threshold which is used to calculate the segment delay upper bound based on the relay sequence. All the information is computed based on the predictable satellite trajectory and dynamic LSN topology. Once decided, the SNO operation center delivers the results to $Src$.


\noindent
\textbf{(2) Intra-segment state probing.} If a path changes or a new session begins, end users, relays and SNO negotiate session keys \cite{SIGCOMM14-pv} which is out of our scope in secure channels \cite{patent-deploy-starlink}. After that, each relay~(including $sat_d$) periodically~(\eg\,tens of seconds) sends a probing packet to its previous relay to probe delay of the current segment. To ensure that probing packets are not tampered with, all nodes on a segment path authenticate the probing packet with adjacent shared symmetric keys \cite{SIGCOMM15-pv}.

\noindent
\textbf{(3) Authentication fields initiation.} Before a packet's departure, $Src$ inserts the corrected sending timestamp, hash value, and relay authentication fields $AUTH$ in the packet header.

\noindent
\textbf{(4) Relay-based segment verification.} Network nodes forward the packets according to the planned path. When a relay receives packets, it verifies the segment path based on its probing ground truth and detour threshold. If packets do not traverse the risk area, the relay computes a MAC value and inserts it with the timestamp to update its part in $AUTH$ field.

\noindent
\textbf{(5) Destination verification.} After receiving a packet, $Dst$ verifies the packet's source (out of our scope), $AUTH$ fields updated by relays, and the last segment path to determine whether the packet bypasses the risk area.

\subsection{Technical challenges}
\label{subsec:technical_challenges}

To accomplish efficient and accurate verification, \name needs to solve the following LSN-specific challenges.

\noindent
\textbf{Appropriate relay selection under LEO dynamics.} Considering the high dynamics in LSNs, correctly and timely selecting verifiable relays away from the risk area while not involving too much additional delay at a global scale is challenging. Specifically, we formulate the \emph{Avoidance-Verifiable Dynamic Relay Selection~(AVDRS)} problem in LSNs, based on which we further design an efficient solution.

\noindent
\textbf{Unpredictable delay jitters.} Only exploiting the period probing delay ground truth as an upper bound to verify whether a packet has gone through the risk area is too strict. It's impossible that the real delays are equal to the probing delays exactly. It's hard to distinguish whether an observed micro delay increase is caused by slight delay jitters or detours which may mislead the segment delay-based verification. 

In the next section, we introduce the details of \name and describe how \name addresses these challenges.


\section{\name Design Details}
\label{sec:design_details}


\subsection{Modeling an LSN}
\label{sec:models}
\noindent
\textbf{Network model.} To model the LEO dynamics, we divide time intervals into discrete time snapshots $\mathcal{T}=\{t_1,t_2,...\}$. During $(t_i,t_{i+1})$, the LSN topology and distance between neighboring satellites are nearly constant. The topology at $t$ is represented by a graph $\mathcal{G}^{\,t}=(\mathcal{V},\mathcal{L}^{\,t})$~($\mathcal{V}$ is the union set of satellites $\mathcal{V_S}$, and ground entities $\mathcal{V_E}$). $\mathcal{L}^{\,t}=\{l_{ij}^{\,t}| i,j\in \mathcal{V}\}$ is the link set of ISLs and GSLs. If node $i, j$ are connected at $t$, the $0\text{-}1$ variable $l_{ij}^{\,t}=1$. The orbit can be represented by a number $p$~($p<\mathcal{P}$, where $\mathcal{P}$ is the amount of orbits). Similarly, we use $n$~($n<\mathcal{H}$, where $\mathcal{H}$ is the number of satellites per orbit) to describe the in-orbit position of a satellite. Formally, a satellite $sat_i$ can be represented by a unique logical coordinate ${sat}_i=(p_i,n_i)$. 

\noindent
\textbf{Link establishment and path construction.} At time $t$, the SNO will compute the risk satellites, \ie\,$\mathcal{RS}^{\,t}=\{{rs}_1,{rs}_2,...\}^{\,t}$ in the risk area $RA$. We set ${path}_{ab}^{\,t}$ as the path at $t$ between nodes $a,b\in\mathcal{V}$ consisting of a sequence of on-path nodes. The segment path between adjacent nodes in ${sat}_s^{\,t}\cup \mathcal{RE}^{\,t}\cup{sat}_d^{\,t}$ is the shortest routing path computed by \emph{Dijkstra} ($\mathcal{RE}^{\,t}$ is the satellite relay set, $\mathcal{RE}^{\,t}=\{r_1,r_2,...\}^{\,t}\subseteq \mathcal{V_S}\setminus (\mathcal{RS}^{\,t}\cup sat_s^{\,t}\cup sat_d^{\,t})$). As the chosen relays are in order, we use a sequence set $\mathcal{Y}^{\,t}=\{y_1,y_2,...\}^{\,t}$ to describe the relation in a path, and $y_i<y_j$ means node $i$ precedes node $j$. The active link between adjacent nodes $i,j\in {path}_{ab}^{\,t}$ is $\tilde{l}_{ij}^{\,t}\le {l}_{ij}^{\,t}$. The delay between any two neighboring nodes $d_{ij}^{\,t}$ can be estimated by dividing their straight-line distance by the light speed. $D_{ab}^{\,t}$ is the total delay of ${path}_{ab}^{\,t}$. Since there exist several paths between two nodes with similar delays, we use a set $\mathcal{EP}_{ab}^{\,t}$ to contain all the related nodes in these paths. 


\subsection{Dynamic relay selection}
\label{subsec:dynamic_relay_selection}

\noindent
\paragraph{AVDRS Formulation}
considering LSN topology's frequent variation, to ease the burden brought by global users' frequent requests, SNO should reduce communication frequency by pre-computing results during a future period for satellite terminals (\eg\, dish) to store them. It's vital to ensure a short end-to-end delay, which limits us from selecting remote relays. We formulate the \emph{Avoidance-Verifiable Dynamic Relay Selection~(AVDRS)} problem which aims at selecting viable satellite relays while satisfying low delay inflation.

\noindent
\textbf{Input:} the topology $\mathcal{L}^{\,t}$, access satellite pairs $sat_s^{\,t}$, $sat_d^{\,t}$ of $Src$ and $Dst$ respectively, and the risk satellites $\mathcal{RS}^{\,t}$. 

\noindent
\textbf{Output:} the selected relay sequence $\mathcal{RE}^{\,t}$. 

\noindent
\textbf{Goals:} \name targets at obtaining the relays while minimizing the end-to-end delay between $Src$ to $Dst$ at $t$:
\begin{equation}
\min{D_{sd}^{\,t}=\sum_{i,j\in{path}_{sd}^{\,t},y_i+1=y_j}{\tilde{l}_{ij}^{\,t}\cdot d_{ij}^{\,t}}} 		
\label{obj}
\end{equation}

\noindent
\textbf{Constraints:} Constraint~\ref{st1} describes the node sequence of a path. $Dis^{\,t}(i,j)$ is the distance between an on-path node and a risk satellite and it limits the candidate relays' range. The larger the $\theta$ is, the further the relay is, and it's more secure while incurring longer end-to-end delay. Constraint~\ref{st3} depicts the connectivity of the constructed path. Constraint~\ref{st4} ensures no equal paths can traverse $RA$. Every relay needs to authenticate the packet and the related field length increases as its number $\sigma$ grows. So the last constraint limits the number of chosen relays to lower the communication overhead. 
Intuitively, more relays mean finer verification granularity, \eg\,all the on-path nodes authenticate. However, from our analysis, a bad relay may incur a longer delay and not contribute to the verification accuracy. 
\begin{equation}\label{st1}
y_i<y_j,\ i,j\in{path}_{sd}^{\,t},\ \mathrm{if}\ i\ \mathrm{precedes}\ j
\end{equation}
\begin{equation}\label{st2}
\min_{i\in{path}_{sd}^{\,t},j\in\mathcal{RS}^{\,t}}{Dis^{\,t}(i,j)}\ge \theta
\end{equation}
\begin{equation}\label{st3}
\tilde{l}_{ij}^{\,t}\le {l}_{ij}^{\,t},\ i,j\in{path}_{sd}^{\,t},\ y_i+1=y_j
\end{equation}
\begin{equation}\label{st4}
\mathcal{EP}_{ij}^{\,t}\cap \mathcal{RS}^{\,t}=\varnothing,\ i,j\in{sat}_s^{\,t}\cup\mathcal{RE}^{\,t}\cup{sat}_d^{\,t}
\end{equation}
\begin{equation}\label{st5}
|\mathcal{RE}^{\,t}|\le \sigma
\end{equation}

\noindent
\textbf{Problem analysis.} To solve the problem above, SNO should calculate the optimal solution in each time slot for large amounts of users. Since the available relay amount for LSNs is $\psi\sim{10}^{3}$, the combination number is~$O(A_{\psi}^{\sigma}$) at a time~(the sequence is in order), which cannot be calculated in a short time, let alone for millions of global users. Besides, to satisfy Constraints~\ref{st1} and \ref{st4}, the problem is more complex as it can be analogized to NP-hard Hamiltonian Path Problem~(HPP)~\cite{np-hard}.

\begin{corollary}
	AVDRS problem can be analogized to HPP.
\end{corollary}
\begin{proof}
HPP specifies a set of points and finds a path in the graph that visits every point only once. We set an unordered $\mathcal{RE}^{\,t}$ and a graph where if the closed $\mathcal{EP}_{ij}^{\,t}\cap\mathcal{RS}^{\,t}=\varnothing,\,i,j\in{sat}_s^{\,t}\cup\mathcal{RE}^{\,t}\cup{sat}_d^{\,t}$, then $l_{ij}^{\,t}=1$. The problem is reduced to finding a sequence that can traverse every relay only once. 
\end{proof}
\paragraph{Dynamic Relay Selection Algorithm~(DRSA)}
\label{sebsec:alg}
The details of DRSA are shown in Alg.\ref{alg:DRSA}. DRSA limits $\mathcal{RS}^{\,t}$ within a range to ensure that as many close low-risk satellites outside the range as possible can be candidate relays. 
\begin{figure}[tb]
	\centering
	\includegraphics[width=0.94\linewidth]{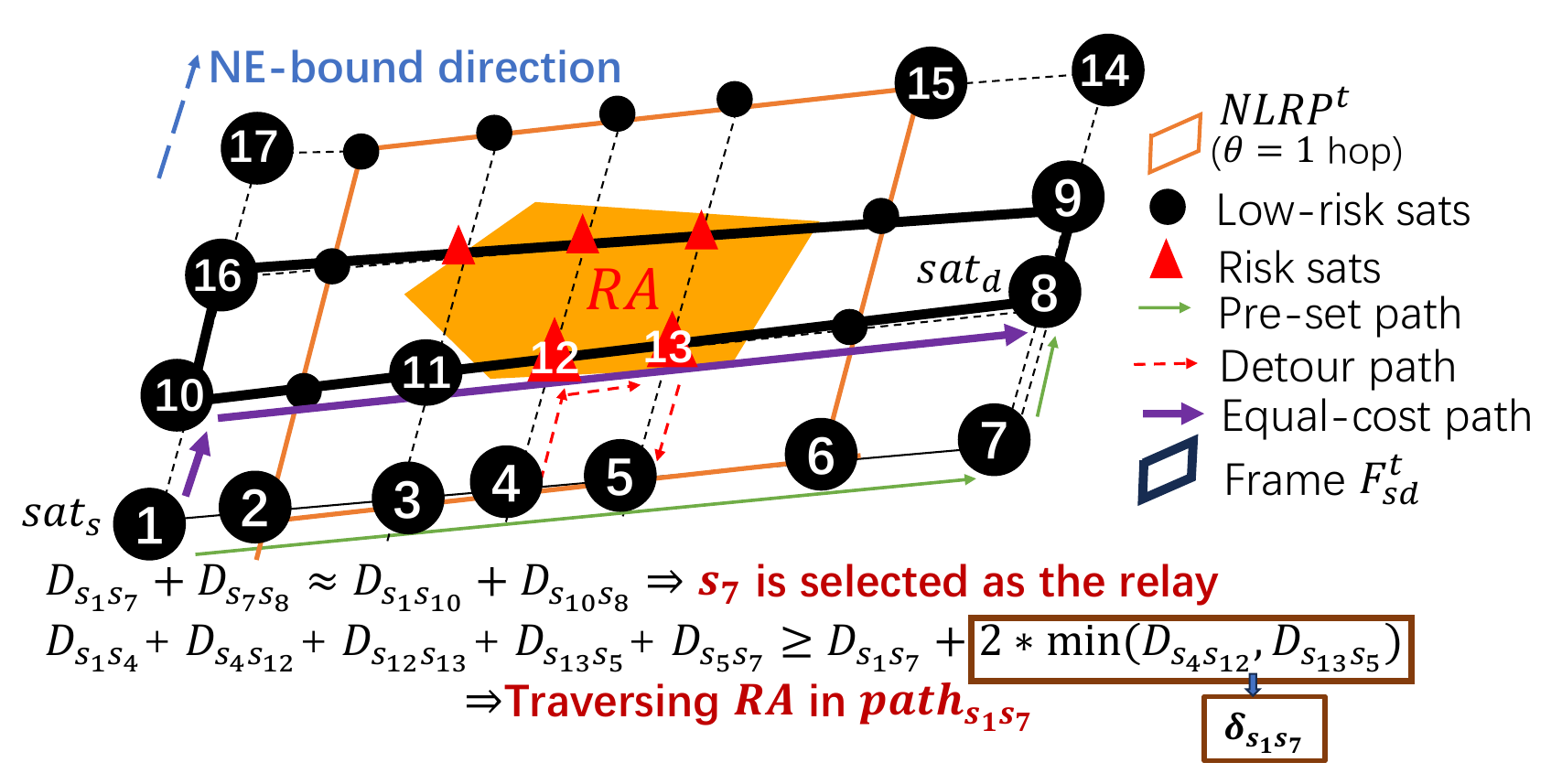}
	\caption{Illustration of DRSA workflow.}
	\vspace{-0.25in}
	\label{fig:nlrp}
\end{figure}

\noindent
\textbf{Calculating the Nearest Low-Risk Planes~(NLRP).}
$NLRP$ consists of eight low-risk planes, \ie\,four planes in each orbital direction. NE-bound direction is that satellites fly from Southwest to Northeast, and SE-bound is from Northwest to Southeast~\cite{starlink_map}. $NLRP$ limits the risk satellites in a range~(it is adjusted by $\theta$ in Constraint~\ref{st2}, which helps set delay buffers against unpredictable delays according to current network conditions. If $\theta$ is defined as the number of hops, $NLRP$ is $\theta$-hop away from the outermost edge risk satellites). Then the upmost, bottommost, leftmost, and rightmost low-risk planes in two orbital directions respectively can be obtained
		\begin{equation}
			NLRP=\{p_{l}^{\,NE},p_{r}^{\,NE},n_{u}^{\,NE},n_{b}^{\,NE},p_{l}^{\,SE},p_{r}^{\,SE},n_{u}^{\,SE},n_{b}^{\,SE}\}\label{eq:NLRP}
		\end{equation}
where, $l, r, u, b$ represent the left, right, up, and bottom, and $NE, SE$ means the orbital direction (Fig.\ref{fig:nlrp} only shows NE-bound $NLRP$ and SE-bound $NLRP$ is a closed frame in different direction interwoven with NE-bound $NLRP$). 


\noindent
\textbf{Dynamic Relay Selection Algorithm (DRSA).} 
Then DRSA selects relays as shown in Alg.\ref{alg:DRSA}. We discuss it in two cases: $dir({sat}_s^{\,t})=dir({sat}_d^{\,t})$ and $dir({sat}_s^{\,t})\ne dir({sat}_d^{\,t})$, where $dir(*)$ is the satellite direction.

When $dir({sat}_s^{\,t})=dir({sat}_d^{\,t})$, firstly a frame $F^{t}_{sd}$ formed by the planes of ${sat}_s^{\,t}$ and ${sat}_d^{\,t}$ is obtained as seen in Fig.~\ref{fig:nlrp}. 
Next, DRSA sets $\mathcal{RE}^{\,t}$ in which a relay $r_i$ is located in the same plane of the previous and the following nodes in $sat_s^{\,t}\cup\mathcal{RE}^{\,t}\cup sat_d^{\,t}$ to make sure there are no other equal-delay segment paths theoretically. Setting relays in this way can reduce the risk that other paths with similar delays traversing $RA$ affect the verification accuracy. If $F_{sd}^{\,t}$ and ${NLRP}^{\,t}$ do not overlap (Line 3-4 in Alg.\ref{alg:DRSA}), \eg\,$F_{s_6s_{14}}^{t}$ formed by $(s_6,s_7,s_{14},s_{15})$ and $NLRP^{\,t}$ in Fig.~\ref{fig:nlrp}, all the nodes constructing possible equal-cost shortest paths from ${s}_6$ to ${s}_{14}$ are contained in $F_{s_6s_{14}}^{\,t}$. In this case, $s_7$ is the only relay since it's farther away from $NLRP^{\,t}$. Otherwise, there are two types of overlap: \emph{semi-overlap} seen in Line 5-7, (\eg\,$F_{s_{1}s_{8}}^{\,t}$ in Fig.~\ref{fig:nlrp}) and \emph{fully-overlap} seen in Line 8-11, (\eg\,$F_{s_{10}s_{9}}^{\,t}$ in Fig.~\ref{fig:nlrp}). In semi-overlap, there exists a candidate path that does not traverse $RA$ while achieving the shortest delay and DRSA selects this corner node as the only relay, \ie\,the equal-cost path depicted by the thicker purple solid arrows traverses $RA$, so $s_7$ is the relay and the shortest path is the path depicted by the finer green solid arrows. In the second case where all the candidate paths contain risk satellites, DRSA enlarges the $F_{sd}^{\,t}$ and selects the corner nodes that are located on the intersection of a certain border of $NLRP^{\,t}$ and the planes of $sat_s^{\,t}$, $sat_d^{\,t}$ as relays, \ie\,$\mathcal{RE}^{\,t}=(s_1,s_7)$.

\begin{figure}[tb]
	\centering
	\setlength{\abovecaptionskip}{-1pt}
	\includegraphics[width=0.9\linewidth]{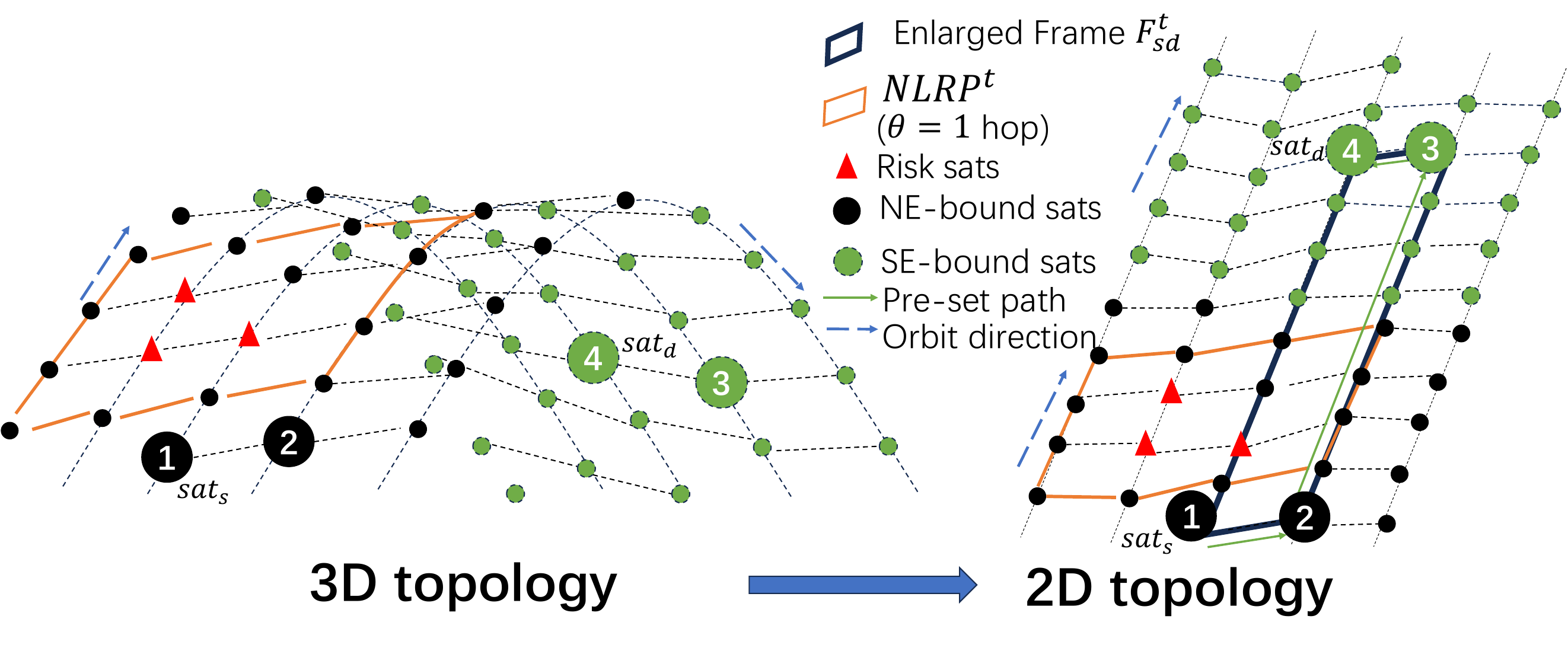}
	\caption{${sat}_s$ and ${sat}_d$ are in different directions.}
	\label{fig:diff_direction_demo}
\end{figure}


If $dir({sat}_s^{\,t})\ne dir({sat}_d^{\,t})$, it can also use the same-direction way above. As shown in Fig.~\ref{fig:diff_direction_demo}, we approximately convert 3-D into 2-D topology and determine the intersection relationship between $F_{s_1s_4}^{\,t}$ and ${NLRP}^{\,t}$. For example, Fig.~\ref{fig:diff_direction_demo} plots the topology transformation and the chosen relays $(s_2,s_3)$ which are the corner nodes in this case.

\noindent
\textbf{Setting detour threshold}. As mentioned in \S\ref{sec:key ideas}, \name divides a path into several segments and each relay or $Dst$ authenticates one segment. The real delays vary dynamically, nearly fluctuating within a certain range. As each segment path is the shortest path, we assume no attacks will incur a lower delay than the probing delay during a period, so \emph{what is the delay upper bound for a packet to be considered as not traversing $RA$?} If an on-path node mis-forwards the packet to $RA$, the minimum detour delay is more than twice the minimum propagation delay between any on-path node and any risk satellite like $\delta_{s_1s_7}$ in Fig.~\ref{fig:nlrp}. Therefore, SNO sets a segment detour delay threshold $\delta_{ij}^{\,t}$ between any two adjacent nodes $i,j$ constructing a segment as:
\begin{equation}\label{eq:detour threshold}\begin{aligned}
\delta_{ij}^{\,t}= \min_{a\in{path}_{ij}^{\,t},b\in\mathcal{RS}^{\,t}}{2\cdot D_{ab}^{\,t}},\\
i,j\in{sat}_s^{\,t}\cup\mathcal{RS}^{\,t}\cup{sat}_d^{\,t},\,y_i<y_j
\end{aligned}
\end{equation}
\setlength{\textfloatsep}{0.12cm}
\begin{algorithm}[t]
	\renewcommand{\algorithmicrequire}{\textbf{Input:}}
	\renewcommand{\algorithmicensure}{\textbf{Output:}}
	\caption{DRSA: Dynamic Relay Selection Algorithm}
	\label{alg:DRSA}
	\begin{algorithmic}[1]
		\REQUIRE $sat_s^{\,t}=(p_s,n_s)$, $sat_d^{\,t}=(p_d,n_d)$, $RA$
		\ENSURE $\mathcal{RE}^{\,t},\Delta^{\,t}$
		\STATE \small{\textbf{get} ${NLRP}^{\,t}$,$F_{sd}^{\,t}$,$\mathcal{RE}^{\,t}=[],\Delta^{\,t}=[]$--detour threshold list}
		\STATE $r_*=(p_s,n_d)$ or $(p_d,n_s)$
		\IF {$F_{sd}^{\,t}\cap {NLRP}^{\,t}==\varnothing$}
			\STATE $\Delta^{\,t} .add(\delta_{sr_*}, \delta_{r_*d})$, $\mathcal{RE}^{\,t}.add(r_*)$
		\ELSIF{$F_{sd}^{\,t}\ and\ {NLRP}^{\,t}\ is\ semi\text{-}overlap$}
			
			\STATE ${path}_{sd}^{\,t}=Dijkstra(sat_s^{\,t},r_*)\cup Dijkstra(r_*,sat_d^{\,t})$
			\STATE $\mathcal{RE}^{\,t}.add(r_*)$, $\Delta^{\,t} .add(\delta_{sr_*},\delta_{r_*d})$ where $r_*\in path_{sd}^{\,t}$ satisfies ${path}_{sd}^{\,t}\cap \mathcal{RS}^{\,t}==\varnothing$

		\ELSIF{$F_{sd}^{\,t}\ and\ {NLRP}^{\,t}\ is\ fully\text{-}overlap$}
			\STATE \textbf{get} $Frame=F_{sd}^{\,t}\cup{NLRP}^{\,t}$
			\STATE \textbf{get} two corner nodes as relays $r_1,r_2$, 
			\STATE $\mathcal{RE}^{\,t}.add(r_1,r_2)$, $\Delta^{\,t}.add(\delta_{sr_1},\delta_{r_1r_2},\delta_{r_2d})$
		\ELSE \RETURN Error /*\eg\,$sat_s^{\,t},sat_d^{\,t}\in \mathcal{RS}^{\,t}$*/
		\ENDIF
		\RETURN $\mathcal{RE}^{\,t},\Delta^{\,t}$
	\end{algorithmic}
\end{algorithm}

\noindent
\textbf{Discussion of number of relays $\sigma$.}
Taking ${sat}_s=s_1,{sat}_d=s_7$ in Fig.~\ref{fig:nlrp} as an example: (i) if we select a bad relay sequence like $(s_{11},s_{3})$, not only is path longer but the relay $s_{11}$ is closer to the risk satellite $s_{12}$, increasing verification inaccuracy; (ii) if we set $s_4$ as the only relay, the verification granularity seems finer as the path is divided into $s_1\to s_4$ and $s_4\to s_7$. But the gain is trivial since the detour threshold in such two segment paths is similar, \ie\ $2\cdot D_{s_{4}s_{12}}\approx 2\cdot D_{s_{5}s_{13}}$. Particularly, no relay is needed if $sat_s$ and $sat_d$ are in the same plane and not in the \emph{fully-overlap} case (\eg\,$s_1$ and $s_7$ are in the same plane). To sum up, DRSA minimizes $\sigma$ to a large extent by constructing $NLRP$ and setting at most two relays. 
\subsection{Pre-process at $Src$}
Before communication begins, $Src$ and SNO exchange \emph{Veri infos} including the detour threshold list $\Delta$, and relays during a future period. Every time a connection begins or relays change (inspected by querying its local cache in each time slot), $Src$, $Dst$, and relays negotiate session keys and achieve the path authorization \cite{SIGCOMM14-pv} (out of our scope) through secure channels between them and SNO \cite{patent-deploy-starlink}. Then, when sending a packet, $Src$ first sends a probe to test the current access delay $D_{access}$ to eliminate the error brought by the first hop. Next, $Src$ records the corrected sending timestamp ${ts}_0={ts}_{send}+D_{access}+\alpha$ ($\alpha$ is processing time) and calculates the truncation of the hash value $HASH=H(Payload||PATH||\Delta^{\,ts_0}||ts_0||\mathcal{RE}^{\,ts_0})[0:l]$, where $[0:l]$ is the truncation of the lowest $l$B of the data to improve goodput~\cite{sec20-pv,ToN23-pv}. Then $Src$ constructs the reserved $AUTH$ chain, makes a final packet source authorization and sends $pkt$. Each relay's $AUTH$ field is shown in Fig.~\ref{fig:auth fields} where $l=4$, including relay's ID $r_i$,  segment detour threshold $\delta_{ij}$, timestamp $ts_i$ when sending out the packet, and its MAC authentication value.
\begin{algorithm}[t]
	\renewcommand{\algorithmicrequire}{\textbf{Input:}}
	\renewcommand{\algorithmicensure}{\textbf{Output:}}
	\caption{Verification Processes}
	\label{alg:process}
	\begin{algorithmic}[1]
	\STATE \emph{\textbf{Step (I): Pre-processing at $Src$}}
		\STATE Probe access delay $D_{access}$
		\STATE  \small{$ts_0 = ts_{send} + D_{access}+\alpha$ /*correct the sending timestamp*/}
		\STATE Perform $HASH$, inquire $PATH$ at $ts_0$
		\STATE Construct $AUTH$ fields and send $pkt$
	\STATE \emph{\textbf{Step (II): Relay-based segment verification}}
		\STATE \textbf{get} $PATH$, $AUTH$, $\mathcal{RE}$ from $pkt$
		\IF {Current ${sat}_i \notin \mathcal{RE}\cup {sat}_d$}
			\STATE Forward $pkt$ or process probing packets
		\ELSIF {Current ${sat}_i \in \mathcal{RE}$}
			\STATE Probe $dt_{r_{i-1}r_i}$ periodically
			\STATE \small{Check if ${AUTH}_{i-1}$ has been updated \textbf{(1)}}
			\STATE \small{Check if $ts_i - ts_{i-1} \le \delta_{r_{i-1}r_i}+dt_{r_{i-1}r_i}$ \textbf{(2)}}
			\IF {(1) and (2) are satisfied}
				\STATE Perform $MAC_{K_{r_i}}[0:l]$,\, update ${AUTH}_i$, forward $pkt$
			
			\ENDIF
			
		\ELSE[/*current $sat_i$ is the dst sat*/]
			\STATE Insert $ts_d$, forward $pkt$ and probe the last segment delay periodically
		\ENDIF
	\STATE \emph{\textbf{Step~(III): Final verification at $Dst$}}
	\STATE \textbf{get} $AUTH$ from $pkt$
	\STATE Verify $pkt$'s source authorization \textbf{(1)}
	\STATE Verify $HASH$ \textbf{(2)}
	\STATE Authenticate MAC values in $AUTH$ \textbf{(3)}
	\STATE \textbf{get} the last relay's ${ts}_{n}\in AUTH_n$, $\delta_{r_{n}d}$
	\STATE Check if $ts_d - ts_n \le \delta_{r_{n}{d}}+dt_{r_{n}d}$ \textbf{(4)}
	\IF {(1)-(4) are all satisfied}
		\STATE Accept $pkt$
	\ELSE
		\STATE Drop $pkt$
	\ENDIF
	\end{algorithmic}
\end{algorithm}
\vspace{-0.1in}
\subsection{Relay-based segment verification}
When receiving a packet, the satellite should determine if it is a relay by checking the $AUTH$ chain. If not, it just forwards it without doing any operations. 

Otherwise, after negotiating the session keys and obtaining the segment path and relay sequence, periodically, the relay probes with a nonce to its previous relay for some link states like queuing and processing delay. The probing reply packets must have been authenticated hop by hop on this segment path through the neighboring shared symmetric keys~\cite{SIGCOMM15-pv,IWQoS18-pv}. Then, the relay obtains the probing delay $dt_{r_{i-1}r_i}$ of the segment path. Similarly, $sat_d$ announces the last segment delay to $Dst$.

Next, upon receiving normal communication packets, the relay $r_i$ just checks if its previous relay $r_{i-1}$ has updated ${AUTH}_{i-1}$ but ignores its validation as \name offloads the step to $Dst$. Then, it determines whether the real delay of this segment path has exceeded the upper bound, \ie\, 
\vspace{-0.07in}
\begin{equation}
\begin{cases}
{ts}_i-{ts}_{i-1}>\delta_{r_{i-1}r_i}+dt_{r_{i-1}r_i},  & \text{ drop the packet }\\
\text{ otherwise, }  & \text{ update}\ AUTH_i
\end{cases}
\end{equation}
If not, it calculates ${MAC}_{K_{r_i}}(HASH||ts_i||r_i)[0:l]$ and inserts it into the packet together with $ts_i$ and sends the packet.
If all the intermediate relays have updated the related fields and the packet is forwarded to ${sat}_d$, ${sat}_d$ inserts the time $ts_d$ when it receives the packet and forwards to $Dst$. 
\subsection{Final path verification at $Dst$}
When $Dst$ receives the packet, it verifies the packet in the following steps: (i) it first authenticates the packet's source validation (this is out of our scope); (ii) after that, it calculates the hash value in the same way as $Src$ to verify if the packet has been tampered; (iii) then it authenticates the MAC chain of relays by re-computing MACs with their shared session keys to verify that the real path has indeed gone through these relays and passed all these relay's verifications; (iv) lastly, $Dst$ authenticates the last segment path in the same way as relay‘s. If all the four conditions are satisfied, $Dst$ accepts it and considers it not traversing the $RA$.

\section{Security Analysis}
\noindent
\textbf{Defense against hijacking.} When benign nodes outside $RA$ mis-forward packets into $RA$ and the packets are redirected to other colluded evil nodes without modifications, the packets are ultimately forwarded to at least a benign node, \eg\ ${sat}_d$. This node will forward them to the next node according to the pre-set $PATH$. Any path inconsistency will incur an extra delay that exceeds the pre-set segment delay upper bounds.

\noindent
\textbf{Defense against path counterfeit and modification.} Some evil nodes may collude with each other by modifying the $PATH$ sequence or the next relay when the packets have entered the $RA$. In the former attack, although the real path is different, $Dst$ will detect it by authenticating the source and recomputing the packet's $HASH$ of the planned path; while in the latter case, a packet may skip some nodes. If skipping some normal nodes on a segment path, the segment delay will be prolonged as mentioned above; and if skipping relays, the downstream relays or $Dst$ will drop the packet as the previous one does not update its $AUTH$ field. If some colluded satellites forge timestamps and MACs, $Dst$ will detect it by re-computing the real relay's MAC chain.

\noindent
\textbf{Defense against replay}. The $HASH$ value contains the timestamp when sending the packet to ensure freshness. Adversaries do not know the session keys, so they cannot forge the source of a packet. Besides, the expired packets will not pass the delay-based verification.

\section{Performance Evaluation}
\label{sec:evaluation}

\begin{figure}[tb]
	\centering
	\includegraphics[height=0.3\linewidth]{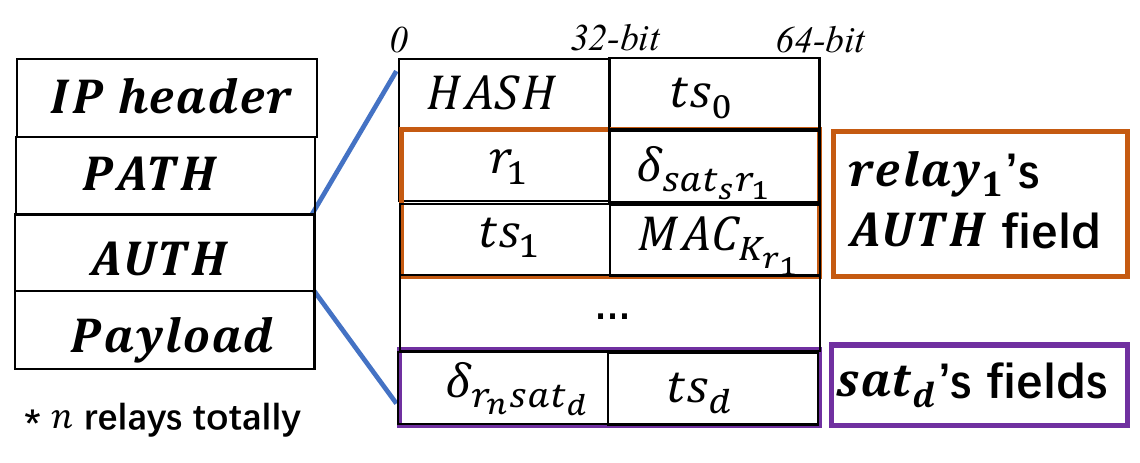}
	\caption{The $AUTH$ header structure of \name.}
	\label{fig:auth fields}
\end{figure}

Our evaluation in this section focuses on the following aspects related to \name: (i) \textbf{Q1:} can \name achieve high accuracy for network path verification in dynamic LSNs? (ii) \textbf{Q2:} can \name's verifiable risk-avoidance routing be exempt from severe delay penalty? (iii) \textbf{Q3:} does \name involve acceptable overhead on satellite routers? and (iv) \textbf{Q4:} can \name scale to large-scale LSNs?

\subsection{Experiment setup}
\noindent
\textbf{Prototype and testbed setups.} 
\name prototype has two major components: \textbf{(i) \name's controller.} \name controller is implemented on a simulation testbed based on StarryNet~\cite{StarryNet}, a novel docker-based framework to simulate large-scaled LSNs written in Python. The controller reads the satellite location data produced by StarryNet first. Then it invokes the \emph{DRSA module} to calculate the risk satellites and $NLRP$s based on the self-defined risk area. After that, the controller calculates the dynamic relays periodically, obtains the segment detour delay thresholds, and constructs the complete routing path for each communication city pair. \textbf{(ii) Satellite routers and city pairs.} They are simulated as individual containers with a complete TCP/IP stack on StarryNet to provide delay with processing time. Based on IPv6 which has an optional \emph{hop-by-hop} header for every intermediate node to process packets, the verification-related data (seen in Fig.~\ref{fig:auth fields}) are embedded in the hop-by-hop header. We use HMAC-SHA256 \cite{sha256} to calculate the digest and extract its first 4 bytes. $Src$ pings $Dst$ continuously. Each ping packet is processed in an individual queue by the module \texttt{Netfilter}~\cite{netfilter}. As seen in Alg.~\ref{alg:process}, after obtaining the segment detour delay thresholds and path from the controller, $Src$ pre-processes each ping packet by embedding them with timestamp, $HASH$, and $AUTH$ fields. The relays update the $AUTH$ fields and forward the packets to $Dst$ which makes the final verification decision. The normal nodes forward the packet without extra operations. The whole simulation environment is built on two high-end servers equipped with Xeon(R) Gold 5215 CPUs (2.50GHz) and 512GB memory.
\noindent
\textbf{Dataset and parameter setting.}
We choose 197 communication city pairs distributed mainly in America and part of South Asia within the service range of Starlink~\cite{servicemap}. To measure the overheads and accuracy, we ping $6000s$~(longer than an orbital period) for some city pairs with different path lengths and numbers of relays simultaneously on StarryNet. In \name, we set three cases of $\theta=1/2/3$ that measure the distance from the $NLRP$ to the risk satellites in Constraint \ref{st2}. We apply Alibi's relay selection methodology for LSNs~(which is also used in \S\ref{subsec:limitation}) and select ground stations as relays. Alibi has a similar user-configurable variable $f=0/0.5$ to find relays. A larger $f$ means fewer relays can be found since the target area where relays are located is smaller, but it's better to resist delay jitters incurred by congestion, \etc\,We set 2 different-sized countries as risk areas: Egypt and North Korea.

\noindent
\textbf{Constellation parameters.}
Our LSN simulation is based on real-world LEO constellation parameters. Specifically, in addition to the constellation and basic network setups in \S\ref{subsec:limitation}, we also simulate Amazon Kuiper with 1156 satellites evenly distributed in 34 orbital planes~\cite{kuiper}.

\subsection{Verification accuracy}
\label{subsec:verification_accuracy}
We set a traffic hijacking scene where packets are forwarded to the risk nodes to see if the prolonged detour delay can be detected. After running our experiments on StarryNet for several hours, \name's verification accuracy reaches $100\%$ in the cases of different numbers of relays. The real detour delay is larger than the pre-set upper bound, as such, no packet that enters the risk area is considered as not traversing the risk area. In the simulation testbed, Alibi produces $38.7\%$ and $27.3\%$ average FP and FN ratios when the risk area is Egypt. Such poor accuracy of Alibi results from frequent variation of the GSLs between satellites and static relays.

\begin{figure}[t] 
	\centering
	\vspace{-0.2in}
	\subfloat[Egypt.]{
		\begin{minipage}[t]{0.48\linewidth} 
			\centering{
				\includegraphics[width=\textwidth]{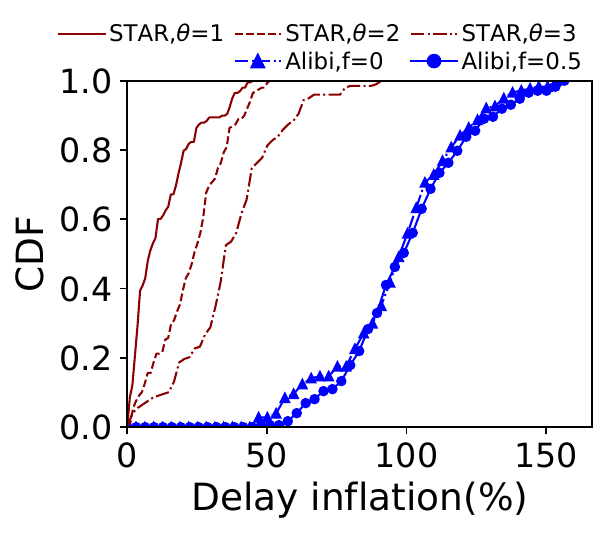}} 
		\end{minipage}%
	}
	\subfloat[North Korea.]{
		\begin{minipage}[t]{0.48\linewidth}
			\centering{
				\includegraphics[width=\textwidth]{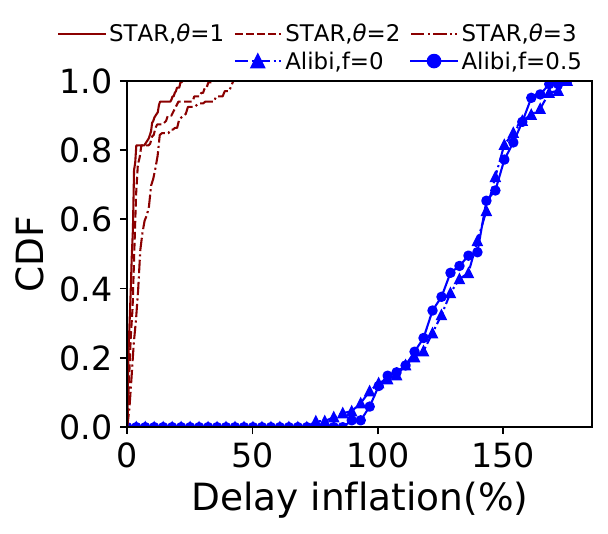}}
		\end{minipage}%
	}
	\caption{Delay inflation of \name~(\textbf{STAR}) and Alibi.}
	\vspace{-0.10in}
	\label{fig:latency inflation}
\end{figure}

\vspace{-0.1in}
\subsection{Network performance} 
\vspace{-0.05in}
Selecting relays far away from the risk area ensures better verification accuracy but incurs larger delay inflation. So we set a configurable variable $\theta$ in Constraint \ref{st2} that depicts the size of $NLRP$ (\ie\,the distance between the $NLRP$ borders and the marginal risk satellites) to make a trade-off between verification accuracy and end-to-end delay. As shown in Fig.~\ref{fig:latency inflation}, a larger $NLRP$ does incur a longer end-to-end delay, which holds true in both constellations as shown in Table~\ref{tab:latency inflation}. However, from a global perspective, the average extra delay of \name is within a reasonable range compared with Alibi. The maximum delay inflation of Alibi reaches $156.54\%$ and $175.49\%$ in two risk countries respectively when $f=0.5$. Besides, limited by the small number and uneven distribution of ground relays, Alibi cannot always find verifiable relays. When $f$ is larger, there is a high possibility that no fixed relays will be found. From our analysis, only 112 of the 197 communication pairs can be assigned a fixed relay when the risk area is North Korea. However, \name can choose global satellites as relays to verify a path only if the users' access satellites are not located inside the $NLRP$.

On top of that, Fig.~\ref{fig:delay thresholds} plots the detour delay thresholds under different $\theta$s defined in Constraint \ref{st2}. \name achieves a delay buffer from tens to hundreds of milliseconds. This additional buffer against unpredictable delay jitters functions better when the $NLRP$ is larger. When avoiding Egypt, the detour threshold increases as the $\theta$ increases. However, there is no explicit detour threshold increase when avoiding North Korea. For example, the average detour thresholds of the three cases of different $\theta$s in Starlink are all about 75ms, which means when the risk area is small, setting a smaller $\theta$ may be enough to achieve a high verification accuracy. 

\begin{table}[t]
\caption{Average delay inflation under different $\theta$s.}
\label{tab:latency inflation}
\vspace{-0.1in}
\begin{center}
\begin{tabular}{|c|c|c|c|c|}
\hline
&\multicolumn{2}{c|}{\textbf{Starlink}}&\multicolumn{2}{c|}{\textbf{Kuiper}} \\
\cline{2-5} 
& \textbf{\textit{Egypt}}& \textbf{\textit{North Korea}}& \textbf{\textit{Egypt}}& \textbf{\textit{North Korea}} \\
\hline
$\theta=1$& $12.44\%$&$3.63\%$&$10.60\%$&$2.37\%$ \\
\hline
$\theta=2$& $23.31\%$&$5.44\%$&$15.67\%$&$3.04\%$ \\
\hline
$\theta=3$& $36.41\%$&$9.01\%$&$24.73\%$&$4.35\%$ \\
\hline
\end{tabular}
\end{center}
\end{table}

\begin{figure}[t] 
	\centering
	\subfloat[Starlink, $RA$ is Egypt.]{
		\begin{minipage}[t]{0.48\linewidth} 
			\centering{
				\includegraphics[width=\textwidth]{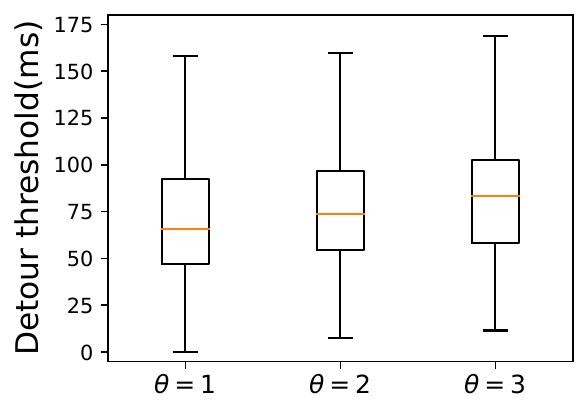}} 
				\vspace{-0.15in}
		\end{minipage}%
	}
	\subfloat[Kuiper, $RA$ is Egypt.]{
		\begin{minipage}[t]{0.48\linewidth}
			\centering{
				\includegraphics[width=\textwidth]{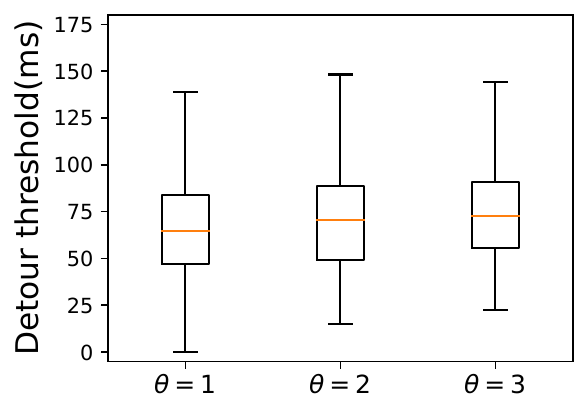}}
				\vspace{-0.15in}
		\end{minipage}%
	}
	\\
	\subfloat[Starlink, $RA$ is North Korea.]{
		\begin{minipage}[t]{0.48\linewidth}
			\centering{
				\vspace{-0.15in}
				\includegraphics[width=\textwidth]{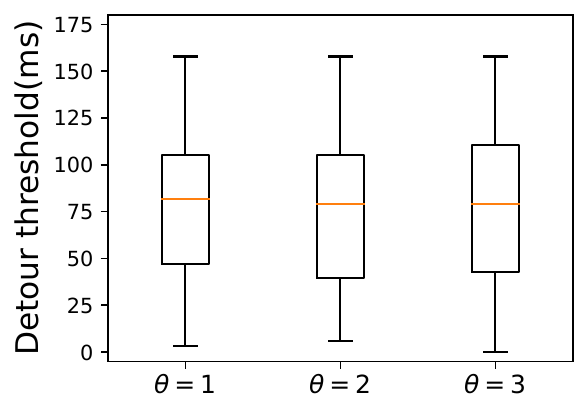}}
				\vspace{-0.15in}
		\end{minipage}%
	}
	\subfloat[Kuiper, $RA$ is North Korea.]{
		\begin{minipage}[t]{0.48\linewidth}
			\centering{
				\vspace{-0.15in}
				\includegraphics[width=\textwidth]{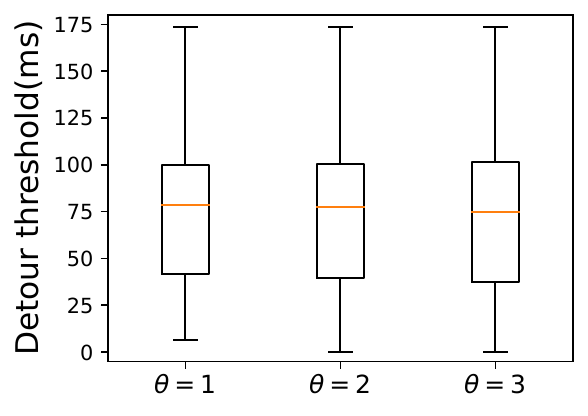}}
				\vspace{-0.15in}
		\end{minipage}%
	}
	\caption{Detour thresholds under different cases.}
	\label{fig:delay thresholds}
\end{figure}

\subsection{Verification overheads}
We analyze the overheads in two aspects: \emph{communication overhead} and \emph{computation overhead} and compare \name with three crypto-based approaches ICING \cite{CoNext11-pv}, OPT \cite{SIGCOMM14-pv}, and EPIC \cite{sec20-pv}. The verification header length of the three methods increases as the path length grows because they verify the path hop by hop.

\noindent
\textbf{Communication overhead.} Our \name field length is at most $48B$ no matter how many hops $N$ a route has as there are two relays at most~(one relay consumes $16B$ fields) as shown in Table~\ref{tb:header lenghth}. To compare them in a meaningful way, we use a $40B$ IPv6 header plus the security-related fields.

As shown in Fig.~\ref{fig:GR}, the longer the path is, the more communication overhead the hop-by-hop methods incur especially ICING. Since the constellation scale is much larger than a single AS in the terrestrial network, it's normal to transmit traffic on a path with tens of hops in LSNs. Fig.~\ref{fig:GR} plots the theoretically maximum goodput ratios~(GR) under different hops. \name achieves 110.16\%, 43.53\%, and 11.33\% higher GRs than ICING, OPT, and EPIC respectively in the case of 30 hops, 1024B payload, and the number of relays $\sigma= 2$.

\begin{table}[t]
\begin{center}
\caption{Security-related field length~(Bytes) in different verification methods.} 
\vspace{-0.1in}
\label{tb:header lenghth} 
\begin{tabular}{|c|c|c|c|c|} 
\hline
\cline{2-5}
\multirow{2}*{\textbf{Path Len $N$}} ~& \textbf{ICING\cite{CoNext11-pv}} &\textbf{OPT\cite{SIGCOMM14-pv}}& \textbf{EPIC\cite{sec20-pv}} & \textbf{\name}\\ 
\cline{2-5}
& $13+42N$ & $52+16N$ & $24+5N$ & 16/32/48\\ 
\hline
$ for\,N=10$ & $433$ & $212$ &$74$ &16/32/48\\
\hline
$ for\,N=20$ & $853$ & $372$ &$124$ & 16/32/48 \\
\hline
$ for\,N=30$ & $1273$ & $532$ &$174$ & 16/32/48\\
\hline
\end{tabular}
\end{center}
\end{table}

\begin{table}[t]
\begin{center}
\caption{Verification delays comparison~($\mu s$). $\sigma$ is the number of relays and $N$ is the path length.} 
\label{tab:process delay} 
\begin{tabular}{|c|c|c|} 
\hline
 &\name &EPIC\\ 
\cline{2-3}
\hline
$\sigma=0,N=35$ & 172$\mu s$ & 6241$\mu s$ \\
\hline
$\sigma=1,N=32$ & 308$\mu s$ & 5820$\mu s$ \\
 \hline
$\sigma=2,N=34$& 550$\mu s$ & 6062$\mu s$ \\

\hline
\end{tabular}
\vspace{-0.2in}
\end{center}
\end{table}

\begin{figure}[t]
    \centering
    \subfloat[20 hops.]{\includegraphics[width=4cm,height=3.3cm]{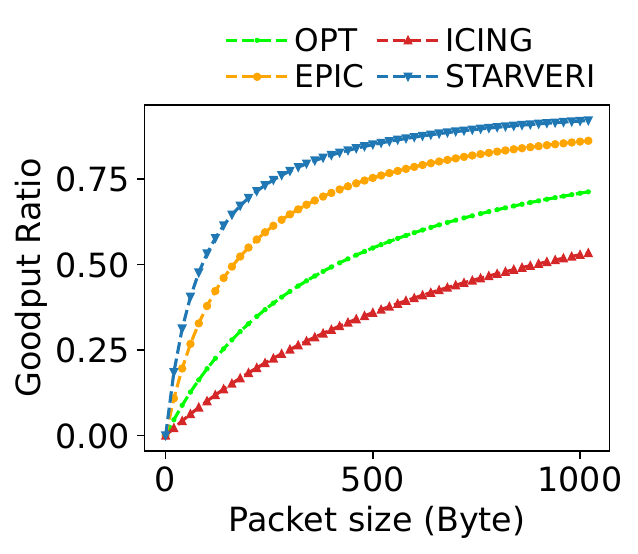}}\quad 
    \subfloat[30 hops.]{\includegraphics[width=4cm,height=3.3cm]{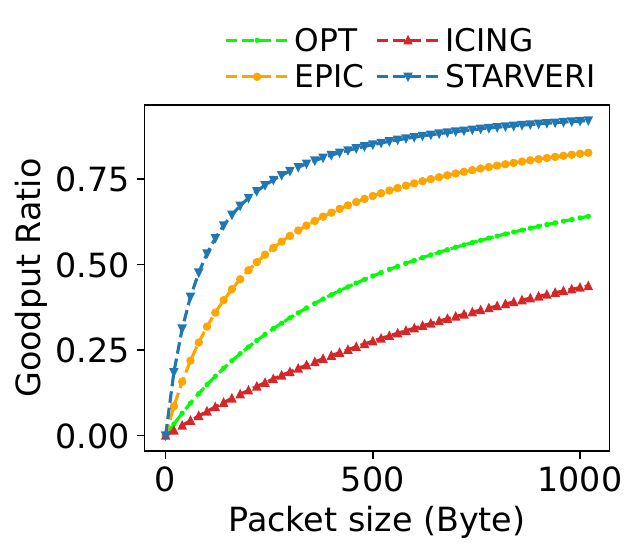}} 
    \vspace{-0.1in}
   
    \caption{Goodput ratios in different hops. (The number of relays of \name $\sigma=2$).}
    \label{fig:GR}
\end{figure}

\noindent
\textbf{Computational overhead.} The complexity of computing relays and detour thresholds is $O(|\mathcal{RS}|\cdot N)$ and there is no need to calculate them frequently. If none of ${sat}_s$, ${sat}_d$ and $NLRP$ changes, the relays will change neither. Here we discuss the variation frequency of $NLRP$. Fig.~\ref{fig:vari frequency} plots the variation frequency of $NLRP$ and the risk satellites $\mathcal{RS}$ every 1000s when the risk area is Egypt. In Starlink, the average variation interval of $NLRP$'s up and bottom planes is about 260.70s. Generally, \emph{the variation frequency of $NLRP$ is fewer than that of $\mathcal{RS}$}. The satellite distribution in Kuiper is sparser, so even if $\mathcal{RS}$ changes, the $NLRP$ may not change. So \name largely decreases relay selection and key negotiation frequency. 

\begin{figure}[t] 
		\subfloat[Starlink.]{
	\centering
		\begin{minipage}[t]{0.48\linewidth} 
				\vspace{-0.2in}
				\includegraphics[width=\textwidth]{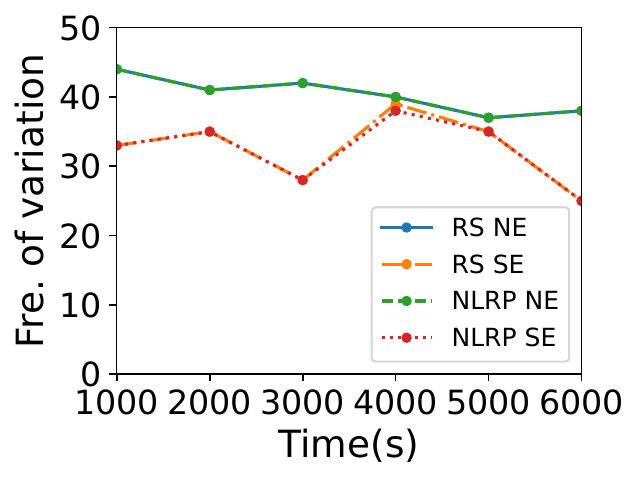}
		\end{minipage}%
	}
	\subfloat[Kuiper.]{
	\centering
		\begin{minipage}[t]{0.48\linewidth}
			\vspace{-0.2in}
				\includegraphics[width=\textwidth]{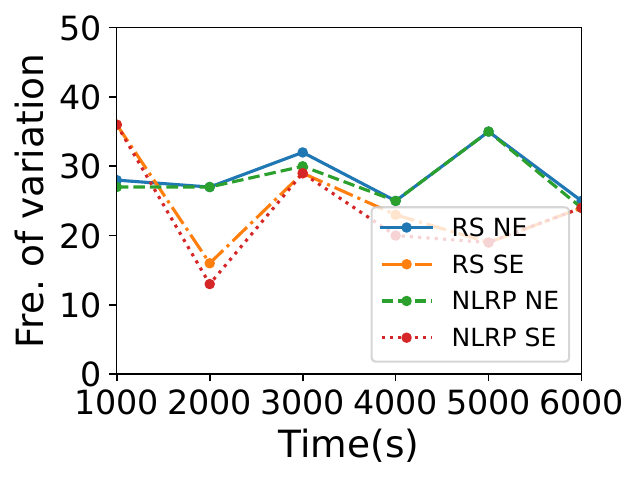}
		
		\end{minipage}
	}
\caption{Variation frequency of $\mathcal{RS}$ and $NLRP$ of Egypt. } 
	\label{fig:vari frequency}
\end{figure}



\begin{figure}[t]
	\centering
	\begin{minipage}[t]{0.47\linewidth} 
		\centering{
			\includegraphics[width=\textwidth]{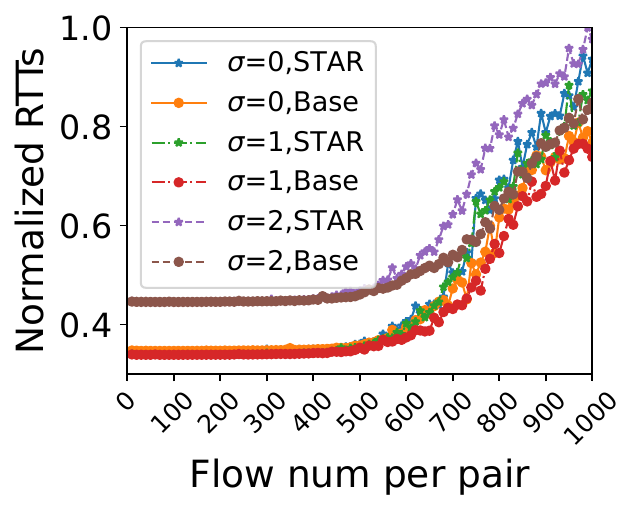}}
		\caption{Normalized RTTs under different flow scales when avoiding Egypt. \small{(\textbf{Base} is baseline without verification.)}}
		\label{fig:flownum}
	\end{minipage}
	\hspace{1mm}
	\begin{minipage}[t]{0.47\linewidth}
		\centering{
			\includegraphics[width=\textwidth]{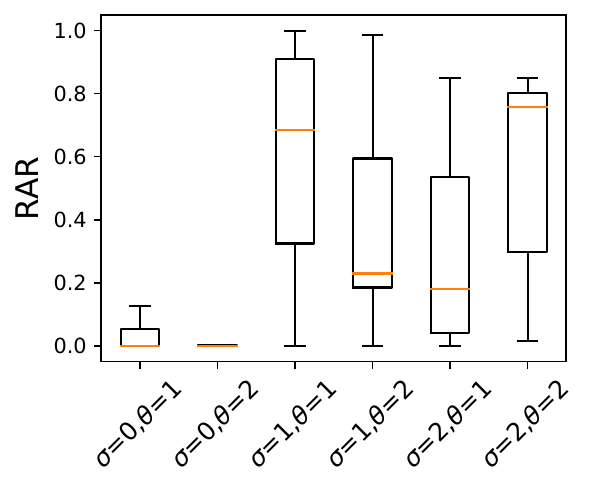}}
		\caption{RARs of all communication pairs in Starlink.}
		\label{fig:RAR}
	\end{minipage}%
	\hspace{0.5mm}
\end{figure}

Besides, \name is not a hop-by-hop verification method, incurring only a little verification delay caused by at most two cryptographic operations. We compare \name to EPIC which also has one MAC operation per hop. As shown in Table~\ref{tab:process delay}, \name's verification delay is 97.2\% less than EPIC when the path has 35 hops.


\subsection{Scalability analysis}
We analyze the RTTs when a path loads different numbers of flows in parallel obtained from StarryNet as shown in Fig.~\ref{fig:flownum}. When the number of flows is small, \name nearly incurs no verification delays. As the number of flows grows, \name incurs more verification delay, especially when $\sigma=2$. The average verification delay of processing between 700 and 1000 flows in parallel is 86ms, 92ms, and 134ms in three $\sigma$s respectively. Hence, SNOs should consider setting a larger $NLRP$ when the number of flows is large.

Next, we analyze the number of required relays. If the number is small, the key negotiation overhead is low. \name will choose a certain number $\sigma$ of relays for a communication pair $cp$ at a certain time. The proportion of these time slots when selecting $\sigma$ relays is $RAR_\sigma^{\,cp}$. For example, \name selects a relay for a pair $cp$ for 600s during the past 1000s, so $RAR_1^{\,cp}=600/1000=0.6$. Fig.~\ref{fig:RAR} plots the boxplot of $RAR_\sigma$s of 197 global pairs in Starlink when avoiding Egypt. The average $RAR_{2}$ of all these pairs when $\theta=1$ is only 29.6\%.

\section{Discussion}
\noindent
\textbf{Probing period.}
If the probing period is too short, more probing packets are needed; and if it's too long, the previously measured delay ground truth may not work if the path has changed. 
Based on our simulations, we quantify the path variation frequency. We set the access strategy that the end users connect to the nearest satellites. The results show that when the risk area is Egypt and $\theta=1$, the paths vary every 26.2s and 26.7s on average in Starlink and Kuiper respectively during the 6000s. The maximum variation interval is 124s and 148s respectively. Normally, a path variation occurs when either of the following cases occurs: i) the access satellites of the end-users vary, or ii) the risk satellites vary. In addition, the path frequency varies with different-sized risk areas. In the future, we will further quantify the modest probing period considering the path variation.


\noindent
\textbf{Real-world deployment.} Since an SNO controls an individual satellite network, the related verification protocols and standards can be deployed in its satellites, and there are encrypted tunnels among the end terminals (\eg\,dish), satellites, and ground stations \cite{patent-deploy-starlink} to transmit secret authentication information. \name can be deployed on the Path-Aware Internet Architecture \cite{springer-SCION,NEBULA,pathlet_routing, SR} where endpoints can customize paths for given destinations, flows, or packets. After the packets are sent out, satellites follow the forwarding paths embedded in the packets.

\noindent
\textbf{Avoidance of multiple risk areas.}
Multiple risk areas may be set by the SNO simultaneously and \name still works by deciding the intersection relations between the frame constructed by the access satellites and the $NLRP$ of each risk area. The SNO can set some middle waypoints between adjacent risk areas. Thus the packets are forwarded to these waypoints in order. Each adjacent waypoint pair (including access satellites) is like an individual city pair's access satellites and the path between them avoids a single risk area. Besides, if these risk areas are close, SNO can aggregate their $NLRP$s as a larger one and then follow the way of \name. So it's possible to have only 2 relays when avoiding two risk areas located closely but the aggregated $NLRP$ will be larger. In future work, we will apply \name to the cases where end users are in the $NLRP$s.
\section{Related Works}
\label{sec:related_works}


We briefly discuss other related works uncovered by \S\ref{subsec:limitation}.

\noindent
\textbf{Path-Aware Networking~(PAN).}
PAN architecture proposes a way for end-users to self-define the paths in data planes flexibly~\cite{springer-SCION,NEBULA}. All these works provide a basis for the source to insert a path into the packet to guide the intermediate nodes to forward it.

\noindent
\textbf{LSN routing.} There are few routing methods aimed at avoiding risk areas in LSNs, so we roughly introduce some routing methods that can be modified to apply in avoiding risk areas. In some solutions~\cite{HotNets18-routing,HotNets19-routing}, the source calculates a low-risk path by utilizing the topology that eliminates risk nodes in advance and inserts it into the header of each packet. Location-based routing~\cite{mobicom24-routing} proposes a distributed geographical routing method that utilizes the information of ground cells and relative locations to decide which is the next hop. LRAR~\cite{iwcnc21-routing} is the only routing method for avoiding risk areas. Though these ways can be used to avoid risk areas, it's hard to prove the real avoidance. 

\noindent
\textbf{Traffic engineering in LSNs.} Some traffic engineering methods~\cite{SIGCOMM07-routing,SIGCOMM10-routing,NSDI13-routing} can transfer the risk nodes into faulty nodes and consider them unreachable. As such, they predict the failure and back up all routing tables in advance. But it brings a lot of storage burden to satellites. StarCure~\cite{INFOCOM23-routing} turns the failure nodes (risk nodes) into congested ones and schedules traffic onto other available routes. But it also lacks the ability to verify that the paths avoid traversing the risk area.

\section{Conclusion}
Dynamic LEO satellites may enter the risk areas and suffer from security issues like hijacking and information stealing in LSNs. To prove that the real route does avoid the risk areas, a path verification method should be proposed in LSNs.

However, existing path verification methods can not adapt to the highly dynamic LSNs for their high communication or computation overheads in crypto-based ways, and low accuracy in delay-based ways. Therefore, we propose the lightweight \name that integrates the advantages of both methods and utilizes satellite trajectories, routing information, and propagation delays to verify that the real path does avoid the risk area. Our extensive simulation proves that \name can achieve near 100\% accuracy. Besides, compared with those hop-by-hop methods, \name largely reduces overheads and has great scalability.

\section{Acknowledgements}
\label{sec:acknowledgement}

We thank our shepherd Olaf Maennel and the anonymous
ICNP reviewers for their comments and suggestions. This work is supported by
the National Key R\&D Program of China (No. 2022YFB3105203) and the National Natural Science Foundation of China (NSFC No.62372259).

\bibliographystyle{IEEEtran}
\bibliography{reference}

\end{document}